\documentclass[a4paper,twoside]{llncs}
\usepackage[T1]{fontenc}
\usepackage[latin9]{inputenc}
\usepackage{array}
\usepackage{verbatim}
\usepackage{amsmath}
\usepackage{amssymb}
\usepackage{graphicx}

\makeatletter

\pdfpageheight\paperheight
\pdfpagewidth\paperwidth

\providecommand{\tabularnewline}{\\}

\newenvironment{lyxlist}[1]
{\begin{list}{}
{\settowidth{\labelwidth}{#1}
 \setlength{\leftmargin}{\labelwidth}
 \addtolength{\leftmargin}{\labelsep}
 }}
{\end{list}}

\@ifundefined{date}{}{\date{}}

\makeatother

\begin{document}

\title{Diffusion, Influence and Best-Response Dynamics in Networks: An Action
Model Approach}

\author{Rasmus K. Rendsvig\\ \email{rendsvig@gmail.com}}

\institute{LUIQ, Lund University, Sweden}
\maketitle
\begin{abstract}
\noindent {\small{}Threshold models and their dynamics may be used
to model the spread of `behaviors' in social networks. Regarding
such from a modal logical perspective, it is shown how standard update
mechanisms may be emulated using action models -- graphs encoding
agents' decision rules. A small class of action models capturing the
possible sets of decision rules suitable for threshold models is identified,
and shown to include models characterizing best-response dynamics
of both coordination and anti-coordination games played on graphs.}%
{\small{} We conclude with further aspects of the action model approach
to threshold dynamics, including broader applicability and logical
aspects. Hereby, new links between social network theory, game theory
and dynamic `epistemic' logic are drawn.}{\small \par}
\end{abstract}

\noindent An individual's choice of phone, language use or convictions
may be influenced by the people around her \cite{Infostorms2013,Skyrms1996,Valente1996}.
How a new trend spreads through a population depends on how agents
are influenced by others, which in turn depends on the structure of
the population and on how easy agents are to influence.

This paper focuses on one particular account of social influence,
the notion of `threshold influence' \cite{Liu2014}. Threshold influence
relies on a simple imitation or conformity pressure effect: agents
adopt a behavior/fashion/semantics whenever some given threshold of
their social network neighbors have adopted it already. So-called
\emph{threshold models}, introduced by \cite{Granovetter1978,Schelling1978},
represent diffusion dynamics under threshold influence. Threshold
models have received much attention in recent literature \cite{Easley_Kleinberg_2010,KempeKleinberg2003,Morris00,Sharakarian_etal_2013},
also from authors in the logic community \cite{Apt2011,christoffhansenlori2013,ZoeRasmus2014,Liu2014,Ruan2011,Girard2011,Zhen2011}.
\enlargethispage{2\baselineskip}

\noindent \smallskip{}

\noindent \hrule\vspace{-7pt}

\noindent \begin{center}
\begin{tabular}{>{\centering}p{12cm}}
\includegraphics[angle=270]{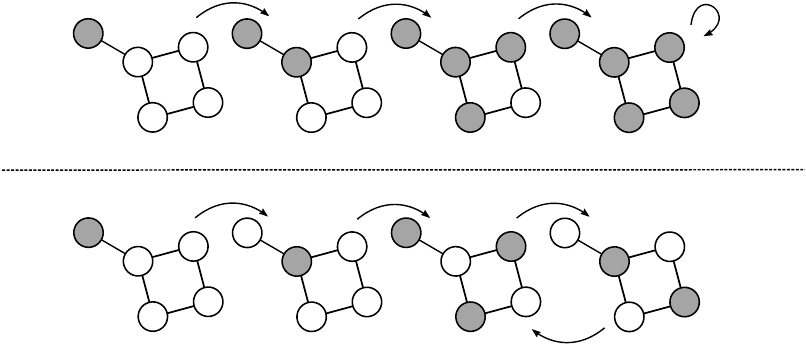}\medskip{}
\tabularnewline
\textbf{\footnotesize{}}%
\begin{tabular}{p{12cm}}
\textbf{\footnotesize{}Fig. 1.}{\footnotesize{} (Definitions below):
A threshold model with 5 agents, threshold $\theta=\frac{1}{4}$,
and behavior $B$ marked by gray. Top: agents change behavior in accordance
with equation (1) and the dynamics reach a fixed point. Bottom: agents
update according to equation (2). Here, the dynamics loop. }\tabularnewline
\end{tabular}\tabularnewline
\end{tabular}
\par\end{center}

In this paper, a novel approach to threshold models is taken by constructing
the dynamics using action models and product update \cite{qqqBaltagBMS_1998,qqqBenthem2006,qqqDitmarsch_Kooi_ontic}.
In this context, an action model may be regarded as a graph that encodes
decision rules. The product of a threshold model and an action model
is again a threshold model, but where each agent has now updated their
behavior according to the encoded decision rules.

The paper progresses as follows. First, threshold models and two typical
update rules are introduced. We then introduce a modal language interpreted
over threshold models, along with action models and product update.
We produce an action model for each of the two introduced update rules,
and show the step-wise equivalence of the two approaches. These two
action models gives rise to a small class of action models, which
is investigated in relation to tie-breaking rules, coordination game
and anti-coordination game best-response dynamics. We conclude with
a discussion of further aspects of the action model approach to threshold
dynamics, including broader applicability and logical aspects.

The motivation for the work is primarily technical. The author found
it interesting that threshold dynamics could so straightforwardly
be encoded using action models. There is however an interesting conceptual
twist: action models are not interpreted as being \emph{informational
events}, but as encoding \emph{decision rules} of agents. Hence, the
class arising from the action model encoding best-responses in coordination
games may be seen as containing all possible sets of decision rules
compatible with agents acting under the used notion of threshold influence.
The class contains variations of tie-breaking rules, and shows a neat
symmetry: for each ``coordination game action model'', the class
contains a ``dual'' version for anti-coordination games. From a
logical perspective, this class is interesting as each arising dynamics
may be treated in a uniform manner, using the reduction axiom method
well-known from dynamic epistemic logic \cite{qqqBaltagBMS_1998,qqqDitmarsch2008}.

\section{Threshold Models and their Dynamics}

\subsubsection*{\noindent Threshold Models. }

\noindent A threshold model includes a network $N$ of agents $\mathcal{A}$
and a behavior $B$ (or fashion, or product, or viral video) distributed
over the agents. As such, it represents the current spread of $B$
through the network. An adoption threshold prescribes how the state
will evolve: agents adopt $B$ when the proportion of their neighbors
who have already adopted it meets the threshold. Formally, a threshold
model is a tuple $\mathcal{M}=(\mathcal{A},N,B,\theta)$ where $\mathcal{A}$
is a finite set of agents, $N\subseteq\mathcal{A\times A}$ a irreflexive
and symmetric network, $B\subseteq\mathcal{A}$ a behavior, and $\theta\in[0,1]$
the adoption threshold.%
\footnote{The literature contains several variations, including infinite networks
\cite{Morris00}, non-inflating behavior \cite{Morris00}, agent-specific
thresholds, non-symmetric relations, weighted links \cite{KempeKleinberg2003},
and multiple behaviors \cite{Apt2011}.%
} For an agent $a\in\mathcal{A}$, her neighborhood is $N(a):=\{b:(a,b)\in N\}$.

\subsubsection*{\noindent Threshold Model Dynamics.}

\noindent Threshold models are used to investigate the spread of a
behavior over discrete time-steps $t_{0},t_{1},...$, i.e., the dynamics
of the behavior. Given an initial threshold model for $t_{0}$, $\mathcal{M}=(\mathcal{A},N,B_{0},\theta)$,
several update policies for the behavior set $B_{0}$ exists.%
\footnote{Attention is here restricted to deterministic, discrete time simultaneous
updates. See e.g. \cite{Newman2002} for stochastic processes.%
} One popular such \cite{ZoeRasmus2014,Easley_Kleinberg_2010,KempeKleinberg2003}
is captured by (\ref{eq:Inflating}):
\begin{equation}
B_{n+1}=B_{n}\cup\left\{ a:\frac{|N(a)\cap B_{n}|}{|N(a)|}\geq\theta\right\} .\label{eq:Inflating}
\end{equation}
I.e., $a$ plays (adopts, follows) $B$ at $t_{n+1}$ iff $a$ plays
$B$$\!$ at $t_{n}$, \textbf{or} a proportion of $a$'s neighbors
\emph{larger }$\!$\emph{or}$\!$\emph{ equal} to the threshold plays
$B$ at $t_{n}$.

The former disjunct makes $B$ \emph{increase} over time, i.e., $\forall n:B_{n}\subseteq B_{n+1}$.
This guarantees that (\ref{eq:Inflating}) reaches a fixed point.
The `or equal to' embeds a tie-breaking rule favoring $B$.

Inflation may be dropped and the tie-breaking rule changed by using
e.g. the policy specified by (\ref{eq:nonInf}) instead:

\begin{equation}
B_{n+1}={\textstyle \left\{ a:\frac{|N(a)\cap B_{n}|}{|N(a)|}>\theta\right\} \cup{\textstyle \left\{ a:\frac{|N(a)\cap B_{n}|}{|N(a)|}=\theta\mbox{ and }a\in B_{n}\right\} .}}\label{eq:nonInf}
\end{equation}

\noindent 

\noindent The second set in the union invokes a conservative tie-breaking
rule: if $\frac{|N(a)\cap B_{n}|}{|N(a)|}=\theta$, $a$ will continue
her behavior from $t_{n}$. That (\ref{eq:nonInf}) does not cause
$B$ to inflate implies the possibility of \emph{loops} in behavior,
i.e. where $B_{n}=B_{n+2}\not=B_{n+1}$. Thereby (\ref{eq:nonInf})
does not necessarily reach a fixed point. %

\subsubsection*{\noindent Threshold Model Dynamics as Induced by Game Play.}

\noindent Threshold influence may naturally be seen as an instance
of a coordination problem: given enough of an agent's neighbors adopt
behavior $B$, the agent will seek to coordinate with that group by
adopting $B$ herself. This coordination problem may be modeled as
a coordination game

\noindent \vspace{-8bp}

\begin{center}
\begin{tabular}{r|c|c|}
\multicolumn{1}{r}{} & \multicolumn{1}{c}{$B$} & \multicolumn{1}{c}{$\neg B$}\tabularnewline
\cline{2-3} 
$B$ & $x,x$ & $0,0$\tabularnewline
\cline{2-3} 
$\neg B$ & $0,0$ & $y,y$\tabularnewline
\cline{2-3} 
\end{tabular}
\par\end{center}

\noindent played on the network: at each time-step, each agent chooses
one strategy from $\{B,\neg B\}$ and plays this strategy against
all their neighbors simultaneously. Agent $a$'s payoff $t_{n}$ is
then the sum of the payoffs of the $|N(a)|$ coordination games that
$a$ plays at time $t_{n}$. With these rules, $B$ is a best-response
for agent $a$ at time $t_{n}$ iff
\begin{equation}
x\cdot{\textstyle \frac{|N(a)\cap B_{n}|}{|N(a)|}}\geq y\cdot{\textstyle \frac{|N(a)\cap\neg B_{n}|}{|N(a)|}}\Leftrightarrow{\textstyle \frac{|N(a)\cap B_{n}|}{|N(a)|}}\geq\frac{y}{x+y}.\label{best-response}
\end{equation}

\noindent Setting $\theta:=\frac{y}{x+y}$, the right-hand side of
(\ref{best-response}) resembles the specifications from (\ref{eq:Inflating})
and (\ref{eq:nonInf}). The precise correlation is that (\ref{eq:nonInf})
captures the best-response dynamics for coordination game play on
networks when using conservative tie-breaking \cite{Morris00}, while
(\ref{eq:Inflating}) captures the same with tie-breaking biased towards
$B$ and the added assumption of a (possibly irrational) `seed'
of agents always playing $B$ \cite{Easley_Kleinberg_2010}.

\section{\noindent Threshold Models, Kripke Models and Action Models}

\noindent A threshold model gives rise to a Kripke model \cite{BlueModalLogic2001}
with $\mathcal{A}$ as domain, $N$ as relation and a valuation $\Vert\cdot\Vert:\Phi\longrightarrow\mathcal{P}(\mathcal{A})$,
$\Phi:=\{B\}$, determining the extension of the $B$ playing agents.
To describe features of agents' neighborhoods, we use a language $\mathcal{L}$
with suitable \emph{threshold modalities}:

\noindent 
\[
\top\;|\; B\;|\;\neg\varphi\;|\;\varphi\wedge\varphi\;|\;\langle{\scriptstyle \leq}\rangle\varphi\;|\;[{\scriptstyle \leq}]\varphi\;|\;({\scriptstyle =})\varphi
\]
The three operators could be parametrized by $\theta$, but to lighten
notation, we leave the threshold implicit.

Intuitively, if $a$ satisfies $\langle{\scriptstyle \leq}\rangle\varphi$,
then there exists a $\theta$ `large enough' set of $a$'s neighbors
that satisfy $\varphi$. E.g., if $\varphi:=B$, then at least a $\theta$
fraction of $a$'s neighbors satisfy $B$. According to (1), $a$
should then change his behavior to $B$. The operator is inspired
by \cite{BaltagTalk,Heifetz1998} and exemplified in Fig. 2. $[{\scriptstyle \leq}]$
is the universal `box' to the existential `diamond' $\langle{\scriptstyle \leq}\rangle$:
if $a$ satisfies $[{\scriptstyle \leq}]$$\varphi$, then all neighbors
in all $\theta$ `large enough' subsets of $a$'s neighborhood
satisfy $\varphi$. Finally, $({\scriptstyle =})\varphi$ captures
that exactly a $\theta$ fraction of the agent's neighbors satisfy
$\varphi$. In particular, if $a$ satisfies $({\scriptstyle =})B$,
then $a$ should invoke a tie-breaking rule.\smallskip{}

With threshold $\theta$, satisfaction in $\mathcal{M}$ is given
by standard Boolean clauses and the following:

\begin{center}
\begin{tabular}{lll}
$\mathcal{M},a\models B$$\phantom{\langle\!\underset{\leq}{\theta}\!\rangle}$ & iff & $\phantom{\langle\!\underset{\leq}{\theta}\!\rangle}$$a\in\Vert B\Vert$\tabularnewline
$\mathcal{M},a\models\langle{\scriptstyle \leq}\rangle\varphi$$\phantom{\langle\!\underset{\leq}{\theta}\!\rangle}$ & iff & $\phantom{\langle\!\underset{\leq}{\theta}\!\rangle}$$\exists C:\theta\leq\frac{|C\cap N(a)|}{|N(a)}\mbox{ and }\forall a\in C,\,\mathcal{M},a\models\varphi$\tabularnewline
$\mathcal{M},a\models[{\scriptstyle \leq}]\varphi$$\phantom{\langle\!\underset{\leq}{\theta}\!\rangle}$ & iff & $\phantom{\langle\!\underset{\leq}{\theta}\!\rangle}$$\forall C:\theta\leq\frac{|C\cap N(a)|}{|N(a)}\mbox{ implies }\forall a\in C,\,\mathcal{M},a\models\varphi$\tabularnewline
$\mathcal{M},a\models({\scriptstyle =})\varphi$$\phantom{\langle\!\underset{\leq}{\theta}\!\rangle}$ & iff & $\phantom{\langle\!\underset{\leq}{\theta}\!\rangle}$$\theta=\frac{N(a)\cap\Vert\varphi\Vert_{\mathcal{M}}}{|N(a)|}.$%
\tabularnewline
\end{tabular}
\par\end{center}

\noindent The extension of $\varphi$ in $\mathcal{M}$ is denoted
$\Vert\varphi\Vert_{\mathcal{M}}:=\{a\in\mathcal{A}:\mathcal{M},a\models\varphi\}$.

\medskip{}

\noindent \hrule

\noindent \begin{center}
\includegraphics[scale=1.2]{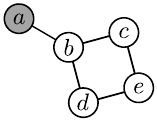}
\par\end{center}

\noindent \begin{center}
\begin{tabular}{p{12cm}}
\textbf{\footnotesize{}Fig. 2.}{\footnotesize{} A threshold model
$\mathcal{M}$ with $\theta=\frac{1}{4}$ and $B$ marked by gray.
$b$ satisfies $\langle{\scriptstyle \leq}\rangle B$, as $\mathcal{M},a\models B$
and $\frac{|\{a\}|}{|\{a,b,c\}|}\geq\frac{1}{4}.$ Agent $e$ satisfies
$[{\scriptstyle \leq}]\neg B$ as $\forall C\subseteq N(e):\frac{|C\cap N(e)|}{|N(e)|}\geq\theta$
(that is, for sets $\{c\},\{d\},\{c,d\}$), $C\subseteq\Vert\neg B\Vert_{\mathcal{M}}$.
Moreover, agent $a$ satisfies $\neg({\scriptstyle =})B\,\wedge\,[{\scriptstyle \leq}]\neg B$
-- hence, according to (\ref{eq:nonInf}), she then should start playing
$\neg B$, whereas (\ref{eq:Inflating}) will not allow her to change.}\tabularnewline
\end{tabular}
\par\end{center}

\noindent \hrule\vspace{5pt}

From $\langle{\scriptstyle \leq}\rangle$, $[{\scriptstyle \leq}]$
and $({\scriptstyle =})$, we define strict versions of the two former.
These are useful when encoding non-biased tie-breaking rules:
\begin{align*}
\langle{\scriptstyle <}\rangle\varphi: & =\langle{\scriptstyle \leq}\rangle\varphi\wedge\neg({\scriptstyle =})\varphi\\{}
[{\scriptstyle <}]\varphi: & =[{\scriptstyle \leq}]\varphi\wedge\neg({\scriptstyle =})\varphi
\end{align*}

Two comments on the threshold operators are due. First, the operators
do not form the basis of a \emph{normal} modal logic: $({\scriptstyle =})$
distributes over neither $\vee$ nor $\wedge$, and the `diamond'
$\langle{\scriptstyle \leq}\rangle$ does not distribute over $\vee$.%
\footnote{The latter was pointed out by Prof. A. Baltag for a similar operator
in \cite{BaltagTalk}.%
} The `box' $[{\scriptstyle \leq}]$ does validate \textbf{K}: $[{\scriptstyle \leq}](\varphi\rightarrow\psi)\rightarrow([{\scriptstyle \leq}]\varphi\rightarrow[{\scriptstyle \leq}]\psi)$
and thus distributes over $\wedge$, but it is not the dual of $\langle{\scriptstyle \leq}\rangle$,
i.e., $[{\scriptstyle \leq}]\varphi\leftrightarrow\neg\langle{\scriptstyle \leq}\rangle\neg\varphi$
is not valid.%
\footnote{The dual of $[{\scriptstyle \leq}]$ would have the universal quantifier
in the semantic clause of $\langle{\scriptstyle \leq}\rangle\varphi$
replaced by an existential one.%
} If $|\mathcal{A}|>1$, the right-to-left direction holds, but not
vice versa. Second, $({\scriptstyle =})\varphi$ does not imply that
$({\scriptstyle =})\neg\varphi$, as the semantics are given w.r.t.
$\theta$. $({\scriptstyle =})\varphi$ does imply that $\frac{N(a)\cap\Vert\neg\varphi\Vert_{\mathcal{M}}}{|N(a)|}=1-\theta$.
This point is important as only $\mathcal{M},a\models({\scriptstyle =})B$,
and not $\mathcal{M},a\models({\scriptstyle =})\neg B$, means that
$a$ must invoke a tie-breaking rule.

\noindent %

\subsubsection*{\noindent Action Models and Product Update.}

Rather than updating threshold models by analyzing best responses
or consulting equations like (\ref{eq:Inflating}) or (\ref{eq:nonInf}),
they may be updated by taking the graph-theoretical product with a
graph that encodes \emph{decision rules}, uniformly followed by all
agents. Such graphs are known as \emph{action models (with postconditions)}
\cite{qqqBaltagBMS_1998,qqqBenthem2006,qqqDitmarsch_Kooi_ontic}.
To illustrate, then (cf. Proposition 1 below) $\mathcal{E}_{1}$ captures
the same dynamics as those invoked by (\ref{eq:Inflating}):

\noindent \begin{center}
\includegraphics[scale=1.8]{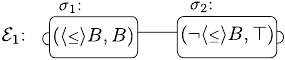}
\par\end{center}

\noindent In the current context, it is natural to interpret each
state of an action models as a decision rule.%
\footnote{The relation between actions is merely a technicality and is not given
an interpretation. Given a re-defined product operation that ignores
the relation of the action model, it could from both a conceptual
and technical point be omitted.%
} E.g., $\sigma_{1}$ encodes the rule `if a $\theta$ fraction or
more of your neighbors play $B$, then play $B$'. State $\sigma_{2}$
encodes that if the agent is not influenced to play $B$, she should
continue her current behavior.

Formally, by an \emph{action model} we here refer to a tuple $\mathcal{E}=(|\mathcal{E}|,R,cond)$
where $|\mathcal{E}|$ is a non-empty domain of states, $R\subseteq|\mathcal{E}|\times|\mathcal{E}|$
is a relation on $|\mathcal{E}|$, and $cond$ a pre- and postcondition
map $cond:|\mathcal{E}|\longrightarrow\mathcal{L}\times\{B,\neg B,\top\}$
with $cond(\sigma)=(\varphi,\psi)=:(pre(\sigma),post(\sigma))$.

The \emph{product update} \cite{qqqBaltagBMS_1998,qqqDitmarsch_Kooi_ontic}
of threshold model $\mathcal{M}$ and action model $\mathcal{E}$
is the threshold model $\mathcal{M}\otimes\mathcal{E}=(\mathcal{A}^{\uparrow},N^{\uparrow},B^{\uparrow},\theta)$
with $\theta$ from \emph{$\mathcal{M}$}, and 
\begin{lyxlist}{00.00.0000}
\item [{$\phantom{.}$}] \noindent $\mathcal{A}^{\uparrow}=\{(a,\sigma)\in\mathcal{A}\times|\mathcal{E}|:\mathcal{M},a\models pre(\sigma)\}$,$\phantom{\langle\!\underset{\leq}{\theta}\!\rangle}$
\item [{$\phantom{.}$}] \noindent $N^{\uparrow}\ni((a,\sigma),(b,\sigma'))$
iff $(a,b)\in N$ and $(\sigma,\sigma')\in R$, and$\phantom{\langle\!\underset{\leq}{\theta}\!\rangle}$
\item [{$\phantom{.}$}] \noindent $B^{\uparrow}\!=\{(s,\sigma)\!:\! s\!\in\! B\wedge post(\sigma)\!\not=\!\neg B\}$$\cup\{(s,\sigma)\!:\! post(\sigma)\!=\!\! B\}$.$\phantom{\langle\!\underset{\leq}{\theta}\!\rangle}$
\end{lyxlist}
By the last condition, $B^{\uparrow}$ consists of 1) the agents in
$B$ minus those who change to $\neg B$, plus 2) the agents that
change to $B$. Hence every agent will after the update again only
play one strategy. If $post(\sigma)=\top$, no change in behavior
is invoked.

\section{\noindent Action Model Dynamics}

\noindent Considering threshold models as Kripke models, it is possible
to construct action models that when applied using product update
will produce model sequences step-wise equivalent to those produced
by (\ref{eq:Inflating}) and (\ref{eq:nonInf}). Moreover, the used
models (in particular $\mathcal{E}_{2}$ below) gives rise to a simple
class of action models. This class, specified below, contains all
natural variations of the decision rules emulating (\ref{eq:Inflating})
and (\ref{eq:nonInf}). Thus, the class specifies all the different
sets of decision rules by which agents may update their behavior while
still behaving in the spirit of present notion of threshold influence.
\begin{proposition}
\noindent For any threshold model $\mathcal{M}$, the action model
$\mathcal{E}_{1}$ applied using product update produces model sequences
step-wise equivalent to those of (\ref{eq:Inflating}).\end{proposition}

\begin{proof}
\noindent Let $\mathcal{M}=(\mathcal{A},N,B,\theta)$ be arbitrary
with (1)-update $(\mathcal{A},N,B^{+},\theta)$ and $\mathcal{E}_{1}$-update
$(\mathcal{A}^{\uparrow},N^{\uparrow},B^{\uparrow},\theta)$. Then
$f$:$a\mapsto(a,\sigma),\sigma\in\{\sigma_{1},\sigma_{2}\}$ is an
isomorphism from $(\mathcal{A},N,B^{+})$ to $(\mathcal{A}^{\uparrow},N^{\uparrow},B^{\uparrow})$.
1) $|\mathcal{A}|=|\mathcal{A}^{\uparrow}|$, as the preconditions
of $\mathcal{E}_{1}$ partitions $\mathcal{A}$ entailing that no
agents multiply or die under product update. 2) $((a,\sigma),(b,\sigma'))\in N^{\uparrow}$
iff $(a,b)\in N$: $R$ from $\mathcal{E}_{1}$ is the full relation,
so $N$ dictates $N^{\uparrow}$. 3) $f(B^{+})=B^{\uparrow}$ as%

\noindent \hspace{1cm}%
\begin{tabular}{>{\centering}p{5cm}c>{\centering}p{3cm}>{\centering}p{3cm}}
\noalign{\vskip0.2cm}
$a\in B^{+}$ & \textbf{iff} &  & \tabularnewline
\noalign{\vskip0.2cm}
$\frac{|N(a)\cap B|}{|N(a)|}\geq\theta$ & \textbf{iff} &  & \tabularnewline
\noalign{\vskip0.2cm}
$\exists C\subseteq N(a)\cap B:\frac{|C|}{|N(a)|}\geq\theta$ & \textbf{iff} &  & \tabularnewline
\noalign{\vskip0.2cm}
$\mathcal{M},a\models\langle{\scriptstyle \leq}\rangle B$ & \textbf{iff} &  & \tabularnewline
\noalign{\vskip0.2cm}
$\mathcal{M},a\models pre(\sigma_{1})$ & \textbf{iff} &  & \tabularnewline
\noalign{\vskip0.2cm}
$\mathcal{M\otimes}\mathcal{E}_{1},(a,\sigma_{1})\models B$ & \textbf{iff} & $f(a)\in B^{\uparrow}$ & \qed\tabularnewline
\end{tabular}
\end{proof}

The action model $\mathcal{E}_{1}$ contains only two states as (\ref{eq:Inflating})
invokes a biased tie-breaking rule, subsumed in the state $\sigma_{1}$
by using the non-strict $\langle{\scriptstyle \leq}\rangle B$ in
the precondition. (\ref{eq:nonInf}), in contrast, invokes a conservative,
unbiased tie-breaking rule. This requires an extra state to encode:

\begin{center}
\includegraphics[scale=1.8]{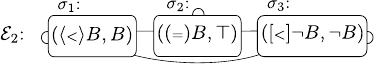}
\par\end{center}

\noindent Interpreted as decision rules, $\sigma_{1}$ of $\mathcal{E}_{2}$
states that if strictly more than a $\theta$ fraction of an agent
$a$'s neighbors plays $B$, then $a$ should do the same; $\sigma_{2}$
embodies the conservative tie-breaking rule: if \emph{exactly }a $\theta$
fraction of $a$'s neighbors play $B$ (and hence a $(1-\theta)$
fraction plays $\neg B$), then $a$ should not change her behavior;
finally, for $\sigma_{3}$, notice that if $[{\scriptstyle <}]\neg B$,
i.e., that all $\theta$ `strictly large enough' subsets of $a$'s
neighbors plays $\neg B$, then there is a strictly larger than $(1-\theta)$
fraction of her neighbors that play $\neg B$---$\sigma_{3}$ states
that in that case, $a$ should also play $\neg B$.
\begin{proposition}
\noindent For any threshold model $\mathcal{M}$, The action model
$\mathcal{E}_{2}$ applied using product update produces model sequences
step-wise equivalent to those of (\ref{eq:nonInf}).\end{proposition}

\begin{proof}
\noindent Analogous to those of Propositions 1 and 3 (see below).\qed
\end{proof}

\subsubsection*{\noindent The Class of Threshold Model Update Action Models.}

\noindent For the reasons mentioned in the proof of Proposition 1,
for an action model to change neither agent set nor network when applied
to an arbitrary threshold model, it must be fully connected and it's
preconditions must form a partition on the agent set. If one further
accepts only preconditions that are in the spirit of standard threshold
model updates, i.e., that agents change behavior based only on the
behavior of their immediate neighbors, then the class of `threshold
model update action models' is easy to map. For by the latter restriction,
$\langle{\scriptstyle <}\rangle B$, $({\scriptstyle =})B$ and $[{\scriptstyle >}]\neg B$
form the unique finest partition%
\footnote{The symmetric variant $\langle{\scriptstyle <}\rangle\neg B$, $({\scriptstyle =})\neg B$
and $[{\scriptstyle >}]B$ is ignored as it is equivalent up to interchange
of $B$ and $\neg B$.%
} on the agent set of any threshold model. Given the three possible
postconditions $B,\top$ and $\neg B$, the class of suitable action
models contains 27 models (Table 1).\medskip{}

\hrule

\begin{center}
\begin{tabular}{cl>{\centering}p{0.6cm}|>{\centering}p{0.6cm}|>{\centering}p{0.6cm}|>{\centering}p{0.6cm}|>{\centering}p{0.6cm}|>{\centering}p{0.6cm}|>{\centering}p{0.6cm}|>{\centering}p{0.6cm}|>{\centering}p{0.6cm}|>{\centering}p{0.6cm}|}
 & \emph{\small{}pre:} &  & {\small{}1} & {\small{}2} & {\small{}3} & {\small{}4} & {\small{}5} & {\small{}6} & {\small{}7} & {\small{}8} & {\small{}9}\tabularnewline
\hline 
{\small{}$\sigma_{1}$:} & {\small{}$\langle{\scriptstyle <}\rangle B$} &  & {\small{}$B$} & {\small{}$B$} & {\small{}$B$} & {\small{}$B$} & {\small{}$B$} & {\small{}$B$} & {\small{}$B$} & {\small{}$B$} & {\small{}$B$}\tabularnewline
{\small{}$\sigma_{2}$:} & {\small{}$({\scriptstyle =})B$} &  & {\small{}$B$} & {\small{}$B$} & {\small{}$B$} & {\small{}$\top$} & {\small{}$\top$} & {\small{}$\top$} & {\small{}$\neg B$} & {\small{}$\neg B$} & {\small{}$\neg B$}\tabularnewline
{\small{}$\sigma_{3}$:} & {\small{}$[{\scriptstyle >}]\neg B$} &  & {\small{}$B$} & {\small{}$\top$} & {\small{}$\neg B$} & {\small{}$B$} & {\small{}$\top$} & {\small{}$\neg B$} & {\small{}$B$} & {\small{}$\top$} & {\small{}$\neg B$}\tabularnewline
 &  & \multicolumn{1}{>{\centering}p{0.6cm}}{} & \multicolumn{1}{>{\centering}p{0.6cm}}{} & \multicolumn{1}{>{\centering}p{0.6cm}}{} & \multicolumn{1}{>{\centering}p{0.6cm}}{} & \multicolumn{1}{>{\centering}p{0.6cm}}{} & \multicolumn{1}{>{\centering}p{0.6cm}}{} & \multicolumn{1}{>{\centering}p{0.6cm}}{} & \multicolumn{1}{>{\centering}p{0.6cm}}{} & \multicolumn{1}{>{\centering}p{0.6cm}}{} & \multicolumn{1}{>{\centering}p{0.6cm}}{}\tabularnewline
 &  &  & {\small{}10} & {\small{}11} & {\small{}12} & {\small{}13} & {\small{}14} & {\small{}15} & {\small{}16} & {\small{}17} & {\small{}18}\tabularnewline
\hline 
{\small{}$\sigma_{1}$:} & {\small{}$\langle{\scriptstyle <}\rangle B$} &  & {\small{}$\top$} & {\small{}$\top$} & {\small{}$\top$} & {\small{}$\top$} & {\small{}$\top$} & {\small{}$\top$} & {\small{}$\top$} & {\small{}$\top$} & {\small{}$\top$}\tabularnewline
{\small{}$\sigma_{2}$:} & {\small{}$({\scriptstyle =})B$} &  & {\small{}$B$} & {\small{}$B$} & {\small{}$B$} & {\small{}$\top$} & {\small{}$\top$} & {\small{}$\top$} & {\small{}$\neg B$} & {\small{}$\neg B$} & {\small{}$\neg B$}\tabularnewline
{\small{}$\sigma_{3}$:} & {\small{}$[{\scriptstyle >}]\neg B$} &  & {\small{}$B$} & {\small{}$\top$} & {\small{}$\neg B$} & {\small{}$B$} & {\small{}$\top$} & {\small{}$\neg B$} & {\small{}$B$} & {\small{}$\top$} & {\small{}$\neg B$}\tabularnewline
 &  & \multicolumn{1}{>{\centering}p{0.6cm}}{} & \multicolumn{1}{>{\centering}p{0.6cm}}{} & \multicolumn{1}{>{\centering}p{0.6cm}}{} & \multicolumn{1}{>{\centering}p{0.6cm}}{} & \multicolumn{1}{>{\centering}p{0.6cm}}{} & \multicolumn{1}{>{\centering}p{0.6cm}}{} & \multicolumn{1}{>{\centering}p{0.6cm}}{} & \multicolumn{1}{>{\centering}p{0.6cm}}{} & \multicolumn{1}{>{\centering}p{0.6cm}}{} & \multicolumn{1}{>{\centering}p{0.6cm}}{}\tabularnewline
 &  &  & {\small{}19} & {\small{}20} & {\small{}21} & {\small{}22} & {\small{}23} & {\small{}24} & {\small{}25} & {\small{}26} & {\small{}27}\tabularnewline
\hline 
{\small{}$\sigma_{1}$:} & {\small{}$\langle{\scriptstyle <}\rangle B$} &  & {\small{}$\neg B$} & {\small{}$\neg B$} & {\small{}$\neg B$} & {\small{}$\neg B$} & {\small{}$\neg B$} & {\small{}$\neg B$} & {\small{}$\neg B$} & {\small{}$\neg B$} & {\small{}$\neg B$}\tabularnewline
{\small{}$\sigma_{2}$:} & {\small{}$({\scriptstyle =})B$} &  & {\small{}$B$} & {\small{}$B$} & {\small{}$B$} & {\small{}$\top$} & {\small{}$\top$} & {\small{}$\top$} & {\small{}$\neg B$} & {\small{}$\neg B$} & {\small{}$\neg B$}\tabularnewline
{\small{}$\sigma_{3}$:} & {\small{}$[{\scriptstyle >}]\neg B$} &  & {\small{}$B$} & {\small{}$\top$} & {\small{}$\neg B$} & {\small{}$B$} & {\small{}$\top$} & {\small{}$\neg B$} & {\small{}$B$} & {\small{}$\top$} & {\small{}$\neg B$}\tabularnewline
\end{tabular}
\par\end{center}

\noindent \textbf{\footnotesize{}Table 1.}{\footnotesize{} Each action
model contains three states with preconditions specified by }\emph{\footnotesize{}pre}{\footnotesize{}
and postconditions by columns 1 to 27.\medskip{}
}{\footnotesize \par}

\hrule\medskip{}

As mention, this class of action models may be seen as containing
all the possible sets of decision rules compatible with the used notion
of threshold influence. Using action models it is a simple, combinatorial
task to map. This is a benefit of using action models to define dynamics
over the set theoretic specification.

\subsubsection*{\noindent Dynamics Induced by Action Models.}

\noindent Note that the action model $\mathcal{E}_{1}$ is not explicitly
listed in Table 1. It is not so as $\mathcal{E}_{1}$ is based on
a coarser partition of the agent set, containing two rather than three
cells. It is however equivalent to the listed model 2: simply collapse
states $\sigma_{1}$ and $\sigma_{2}$ to one.

The class include three trivial dynamics induced by models 1, 14 and
27 and seven that make little sense (4, 7, 8, 16, 17 and 24).

The best-response dynamics of coordination games are emulated by models
3, 6 and 9, capturing discriminating (3,9) and conservative (6) tie-breaking
(cf. Proposition 2), while models 2, 5, 15 and 18 capture inflating
(`seeded') coordination game dynamics.

Proposition 3 below lends credences to the conjecture that models
19, 22 and 25 capture the best-response dynamics for anti-coordination
games with discriminating (19,25) and conservative (22) tie-breaking,
and that 10, 13, 23 and 26 capture inflating dynamics of anti-coordination
games.
\begin{proposition}
\noindent For any threshold model, the best-response dynamics of the
anti-coordination game

\noindent \hspace{5cm}%
\begin{tabular}{r|c|c|}
\multicolumn{1}{r}{} & \multicolumn{1}{c}{$B$} & \multicolumn{1}{c}{$\neg B$}\tabularnewline
\cline{2-3} 
$B$ & $0,0$ & $y,x$\tabularnewline
\cline{2-3} 
$\neg B$ & $x,y$ & $0,0$\tabularnewline
\cline{2-3} 
\end{tabular}\vspace{7pt}

\noindent played with the conservative tie-breaking rule is step-wise
equivalent to applying the action model 22 ($\mathcal{E}_{22}$) from
Table 1 with $\theta=\frac{x}{y+x}$.\end{proposition}

\begin{proof}
\noindent Let $\mathcal{M}=(\mathcal{A},N,B,\theta)$. Playing $B$
is a best-response in $\mathcal{M}$ for agent $a$ iff
\[
y\cdot{\textstyle \frac{|N(a)\cap\neg B|}{|N(a)|}}\geq x\cdot{\textstyle \frac{|N(a)\cap B|}{|N(a)|}}\Leftrightarrow{\textstyle \frac{|N(a)\cap\neg B|}{|N(a)|}}\geq{\scriptstyle \frac{x}{y+x}}=\theta
\]

\noindent Hence, given the tie-breaking rule, the next set of $B$-players
will be
\[
B^{+}=\{a:{\textstyle \frac{|N(a)\cap\neg B|}{|N(a)|}}>\theta\}\cup\{a:{\textstyle \frac{|N(a)\cap\neg B|}{|N(a)|}}=\theta\mbox{ and }a\in B\}.
\]

Let $\mathcal{M}\otimes\mathcal{E}_{22}=(\mathcal{A}^{\uparrow},N^{\uparrow},B_{n}^{\uparrow},\theta)$.
Then $g:a\mapsto(a,\sigma)$, $\sigma\in|\mathcal{E}_{22}|$, is an
isomorphism from $(\mathcal{A},N,B^{+})$ to $(\mathcal{A}^{\uparrow},N^{\uparrow},B_{n}^{\uparrow})$.
That $(\mathcal{A},N)\cong_{g}(\mathcal{A}^{\uparrow},N^{\uparrow})$
follows from the proof of Proposition 1.

\begin{tabular}{>{\centering}p{4cm}c>{\centering}p{4.5cm}cc}
\noalign{\vskip0.2cm}
$a\in B^{+}$ \textbf{iff} &  &  &  & \tabularnewline
\noalign{\vskip0.2cm}
$\frac{|N(a)\cap\neg B|}{|N(a)|}>\theta$ & or & $\frac{|N(a)\cap\neg B|}{|N(a)|}=\theta$ and $a\in B$ & \textbf{iff} & \tabularnewline
\noalign{\vskip0.2cm}
$\mathcal{M},a\models\langle{\scriptstyle <}\rangle\neg B$ & or & $\mathcal{M},a\models B\wedge({\scriptstyle =})\neg B$ & \textbf{iff} & \tabularnewline
\noalign{\vskip0.2cm}
$\mathcal{M},a\models pre(\sigma_{1})$ & or & $\mathcal{M},a\models B\wedge pre(\sigma_{2})$ & \textbf{iff} & \tabularnewline
\noalign{\vskip0.2cm}
$\mathcal{M\otimes E}_{22},(a,\sigma_{1})\models B$\linebreak{}
 (as $post(\sigma_{1})=B$) & or & $\mathcal{M\otimes E}_{22},(a,\sigma_{2})\models B$\linebreak{}
(as $post(\sigma_{2})=\top$) & \textbf{iff} & $g(a)\in B^{\uparrow}$\qed\tabularnewline
\noalign{\vskip0.2cm}
 &  &  &  & \tabularnewline
\end{tabular}
\end{proof}

\subsubsection*{Logics for Threshold Dynamics.}

Given the uniform, action model approach to the dynamics outlined,
it may be conjectured that the dynamics may also be treated by a uniform
logical approach, particularly the reduction axiom method well-known
from dynamic epistemic logic \cite{qqqBaltagBMS_1998,qqqDitmarsch2008}.

Three things are required to obtain a complete logic for one of the
dynamics:
\begin{enumerate}
\item A complete axiomatization for the threshold operators $\langle{\scriptstyle \leq}\rangle,[{\scriptstyle \leq}]$
and $({\scriptstyle =})$,
\item A complete axiomatization of the network properties, and
\item Reduction laws for the used action model.
\end{enumerate}
For 1, one may search for results in the literature on probabilistic
modal logic. No suitable, general result is known to the author. 2
is easily obtained, though it requires a richer language, extending
$\mathcal{L}$ with a normal modal operator $\lozenge$ and hybrid
logical \emph{nominals }$\{i,j,...\}$. The latter is required to
express the irreflexivity of the network relation, characterized by
$i\rightarrow\neg\lozenge i$. To complete a combined logic, interaction
axioms for the thresholds operators and normal modal operators should
also be added. A reduction axiom-based logic for action models with
post-conditions already exists (the logic \textbf{UM }from \cite{qqqDitmarsch_Kooi_ontic}),
but the system should be modified to suit the hybrid nominals and
threshold modalities. If such a combined logic is obtained for one
of the dynamics, one will automatically obtain complete logics for
all of the 27 dynamics induced by the action models of Table 1, with
the only variation between them being the used action model in the
dynamic modalities.

\section{An Action Model for `Belief Change in the Community'}

One reviewer asked whether there is a relation between the action
model approach used here, and the finite state automata approach introduced
in \cite{Zhen2011} for threshold influence dynamics of preferences,
and in particular, whether a translation between the two approaches
exist. We conjecture that this is indeed the case. To lend credence
to this conjecture, we show this may be done for the slightly simpler
framework of threshold influence of belief change from \cite{Liu2014}.

The basic framework of \cite{Liu2014} investigates the dynamics of
\emph{strong }and \emph{weak influence} of beliefs among agents in
a symmetric and irreflexive network. Beliefs are represented by three
mutually exclusive atoms $Bp,B\neg p$ and $Up$, evaluated at agents
in the network, as above. $\mathcal{M},a\models Bp$ reads `$a$
believes $p$', and being undecided about $p$, $Up$, is equivalent
to $\neg Bp\wedge\neg B\neg p$. To describe the network, a normal
box operator $F$ is used: $\mathcal{M},a\models F\varphi$ iff $\forall b\in N(a),\mathcal{M},b\models\varphi$.
$F$ has dual $\langle F\rangle$ -- $\langle F\rangle\varphi$ reads
`I have a friend that satisfies $\varphi$'. Call the language $\mathcal{L}'$.

An agent is strongly influence to believe $\varphi\in\{p,\neg p\}$
if all her friends believe $\varphi$, and weakly influenced to believe
$\varphi$ if no friends believe $\neg\varphi$ while at least one
friend believes $\varphi$. With
\begin{align*}
S\varphi & :=FB\varphi\wedge\langle F\rangle B\varphi\mbox{ and}\\
W\varphi & :=F\neg B\neg\varphi\wedge\langle F\rangle B\varphi,
\end{align*}
the dynamics of strong and weak influence are then characterized by
the finite state automaton in Fig. 3, applied to all agents simultaneously.

\noindent \begin{center}
\begin{tabular}{>{\centering}p{12cm}}
\includegraphics[scale=1.6]{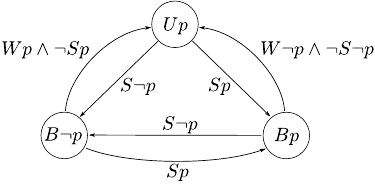}\medskip{}
\tabularnewline
\textbf{\footnotesize{}}%
\begin{tabular}{p{12cm}}
\textbf{\footnotesize{}Fig. 3.}{\footnotesize{} The automaton of \cite{Liu2014},
which characterizes agents' belief change under weak and strong influence.
If an agent is undecided about $p$, i.e., in the state $Up$, and
is strongly influenced to believe $p$, $Sp$, she will change to
state $Bp$, i.e., believe $p$. The automaton is deterministic.}\tabularnewline
\end{tabular}\tabularnewline
\end{tabular}
\par\end{center}

Given this setup, it is no hard task to construct an action model
over $\mathcal{L}'$ that will invoke the same dynamics. This may
be done systematically by the construction: 1) for each state-transition-state
triple $(s,t,s')$ from the automaton, construct an action model state
$\sigma$ with the conjunction of the labels of $s$ and $t$ as precondition,
and the label from $s'$ as postcondition, and 2) let the relation
of the action model be the full relation. The resulting action model
$\mathcal{I}$ is depicted in Fig. 4. It is easy to verify that the
effects of the two approaches are equivalent.

\noindent \begin{center}
\begin{tabular}{>{\centering}p{12cm}}
\includegraphics[scale=1.6]{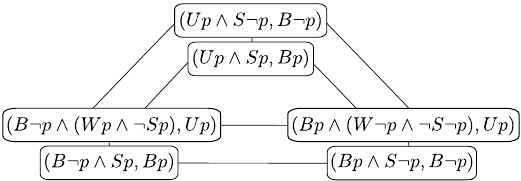}\medskip{}
\tabularnewline
\textbf{\footnotesize{}}%
\begin{tabular}{p{12cm}}
\textbf{\footnotesize{}Fig. 4.}{\footnotesize{} The action model $\mathcal{I}$,
invoking the same dynamics as the automaton of Fig. 3 (some edges
are omitted). The top-most state makes an agent change from state
$Up$ to state $B\neg p$ if the agents also is strongly influenced
to believe $\neg p$, etc.}\tabularnewline
\end{tabular}\tabularnewline
\end{tabular}
\par\end{center}

\noindent The construction method used defines a function from automata
to action models. If one restricts attention to action models with
preconditions of the form $(\varphi\wedge\psi)$, a function from
action models to automata may be defined by the construction: 1) for
each action model state $\sigma$, $cond(\sigma)=((\varphi\wedge\psi),\chi)$,
construct a automaton state with label $\varphi$ and one with label
$\chi$, and collapse all automata states with equivalent labels,
and 2) for each all automaton states with labels $\varphi$ and $\chi$,
add a transition with label $\psi$ between them if there exists an
action model state with $cond(\sigma)=((\varphi\wedge\psi),\chi)$.
Combining the two constructions provides a bijection, serving as translation.%

\subsubsection*{A Logic for Belief Change in the Community.}

Given that it is possible to emulate the dynamics invoked by the finite
state automaton using an action model, finding a sound and complete
logic for the dynamics should be unproblematic. In fact, as $F$ is
a normal modal operator, the case is simpler than for threshold dynamics.
Again some hybrid machinery is required to capture the irreflexive
frame condition, but if this requirement is dropped, the reduction
axiom system from \cite{qqqDitmarsch_Kooi_ontic} provides the desired
result.

\section{\noindent Closing Remarks}

\noindent It has been argued that action models may be used to emulate
the best-response dynamics on coordination and anti-coordination games
played on networks by showing the product updates equivalent to the
threshold model dynamics induced by game play, and that the method
is applicable to the framework of threshold influence from \cite{Liu2014}.
It is conjectured that the action model approach to threshold dynamics
lightens the work of finding complete logics, using methods well-known
from dynamic epistemic logic, hereby providing new connections between
game theory, social network theory and dynamic `epistemic' logic.

Two questions present themselves. First, is it possible to rationalize
the seven unaccounted for action models in the identified class, by
moving from action models to game playing situations? Second, what
is the extent of the applicability of action models? The present paper
utilizes only a fraction of the potential of action models, as such
may also be used to systematically alter the agent set and network.
Changing the agent set may be used to model agent death and birth,
whereby deterministic SIRS-like epidemiological dynamics \cite{Newman2002}
may be captured. Alterations to the social network may be used to
model e.g. rise in popularity of information sources.

\end{document}